\documentclass[a4paper,11pt]{article}
\usepackage[margin=1in]{geometry}

\usepackage[linesnumbered,vlined]{algorithm2e}
\usepackage{amsmath,amssymb,amsthm,thmtools} 
\usepackage{cite}
\usepackage{mathtools} 
\usepackage{enumerate}
\usepackage{hyperref}
\usepackage{booktabs} 

\newcommand{\calC}{\ensuremath{\mathcal{C}}}

\newcommand{\E}{\mathop{\mathbb{E}}}

\newcommand{\bbN}{\mathop{\mathbb{N}}}

\newcommand{\size}[1]{\ensuremath{\left|#1\right|}}
\newcommand{\set}[1]{\left\{ #1 \right\}}

\newcommand{\cgst}{\textsc{congest}}
\newcommand{\lcl}{\textsc{local}}
\newcommand{\sixap}{Algorithm~${\tt 6AP}$}

\DeclareMathOperator{\sent}{\textnormal{\texttt{sent}}}
\DeclareMathOperator{\false}{\textnormal{\texttt{FALSE}}}
\DeclareMathOperator{\true}{\textnormal{\texttt{TRUE}}}

\newcommand{\proxlv}{{\textnormal{\texttt{PL}}}^\ast_v}
\newcommand{\pathmv}{{\textnormal{\texttt{PM}}}^\ast_v}
\DeclareMathOperator{\algprox}{\textnormal{\texttt{PL}}}
\DeclareMathOperator{\algpath}{\textnormal{\texttt{PM}}}

\newcommand{\remove}[1]{}

\declaretheorem{theorem}
\declaretheorem[numberlike=theorem]{lemma}

\declaretheorem[numberlike=theorem]{definition}

\newtheorem*{theorem*}{Theorem}

\newenvironment{theorem-repeat}[1]{\begin{trivlist}
\item[\hspace{\labelsep}{\bf\noindent Theorem \ref{#1}. }]\em }%
{\end{trivlist}}
\newenvironment{lemma-repeat}[1]{\begin{trivlist}
\item[\hspace{\labelsep}{\bf\noindent Lemma \ref{#1}. }]\em }%
{\end{trivlist}}

\title{The Sparsest Additive Spanner via Multiple Weighted BFS Trees}

\author{Keren Censor-Hillel\footnotemark[3]
	\and Ami Paz\footnotemark[4]
	\and Noam Ravid\footnotemark[3]}
\date{}

\footnotetext[3]{Department of Computer Science, Technion, Israel.\\ \texttt{\{ckeren,noamrvd\}@cs.technion.ac.il}.}
\footnotetext[4]{Faculty of Computer Science, University of Vienna, Austria.
\texttt{ami.paz@univie.ac.at}.}

\usepackage{xcolor}

\begin{document}

\maketitle

\begin{abstract}
Spanners are fundamental graph structures that sparsify graphs at the cost of small stretch.
In particular, in recent years,
many sequential algorithms constructing additive all-pairs spanners were designed, providing very sparse small-stretch subgraphs.
Remarkably, it was then shown that the known $(+6)$-spanner constructions are essentially the sparsest possible, that is, larger additive stretch cannot guarantee a sparser spanner, which brought the stretch-sparsity trade-off to its limit.
Distributed constructions of spanners are also abundant. However, for additive spanners, while there were algorithms constructing $(+2)$ and $(+4)$-all-pairs spanners,
the sparsest case of $(+6)$-spanners remained elusive.

We remedy this by designing a new sequential algorithm for constructing a $(+6)$-spanner
with the essentially-optimal sparsity of $\tilde{O}(n^{4/3})$ edges.
We then show a distributed implementation of our algorithm,
answering an open problem in~\cite{Censor-HillelKP18}.

A main ingredient in our distributed algorithm is an efficient construction of multiple weighted BFS trees.
A weighted BFS tree is a BFS tree in a weighted graph,
that consists of the lightest among all shortest paths from the root to each node.
We present a distributed algorithm in the CONGEST model,
that constructs multiple weighted BFS trees in $|S|+D-1$ rounds,
where $S$ is the set of sources and $D$ is the diameter of the network graph.

$ $


\end{abstract}

\remove{---- text only abstract:
	
Spanners are fundamental graph structures that sparsify graphs at the cost of small stretch. In particular, in recent years, many sequential algorithms constructing additive all-pairs spanners were designed, providing very sparse small-stretch subgraphs. Remarkably, it was then shown that the known (+6)-spanner constructions are essentially the sparsest possible, that is, a larger additive stretch cannot guarantee a sparser spanner, which brought the stretch-sparsity trade-off to its limit. Distributed constructions of spanners are also abundant. However, for additive spanners, while there were algorithms constructing (+2) and (+4)-all-pairs spanners, the sparsest case of (+6)-spanners remained elusive.

We remedy this by designing a new sequential algorithm for constructing a (+6)-spanner with the essentially-optimal sparsity of roughly O(n^{4/3}) edges. We then show a distributed implementation of our algorithm, answering an open problem in [Censor-Hillel et al., Distributed Computing 2018].

A main ingredient in our distributed algorithm is an efficient construction of multiple weighted BFS trees. A weighted BFS tree is a BFS tree in a weighted graph, that consists of the lightest among all shortest paths from the root to each node. We present a distributed algorithm in the CONGEST model, that constructs multiple weighted BFS trees in |S|+D-1 rounds, where S is the set of sources and D is the diameter of the network graph.

----}

\newpage

\section{Introduction}
A spanner of a graph $G$ is a spanning subgraph $H$ of $G$
that approximately preserves distances.
Spanners find many applications in distributed computing~\cite{Chechik13a,Censor-HillelHKM12,PelegU89a,PelegU89b,ThorupZ01},
and thus their distributed construction is the center of many research papers.
We focus on spanners that approximately preserve
distances between all pairs of nodes,
and where the stretch is only by an additive factor (\emph{purely-additive all-pairs} spanners).

Out of the abundant research on distributed constructions of spanners, only two papers discuss the construction of \emph{purely additive} spanners in the \cgst{} model:
the construction of $(+2)$-spanners is discussed in~\cite{LenzenP13},
and the construction of $(+4)$-spanners and $(+8)$-spanners in~\cite{Censor-HillelKP18},
along with other types of additive spanners and lower bounds.
However, the distributed construction of $(+6)$-spanners remained elusive,
stated explicitly as an open question in~\cite{Censor-HillelKP18}.
This is especially important since additive factors greater then $6$ cannot yield essentially sparser spanners~\cite{AbboudB17}.

In this paper,
we give a distributed algorithm for constructing a $(+6)$-spanner,
with an optimal number of edges up to sub-polynomial factors;
our spanner is even sparser than the $(+8)$-spanner presented in~\cite{Censor-HillelKP18}.
Several sequential algorithms building $(+6)$-spanners were presented,
but none of them seems to be appropriate for a distributed setting.
Thus, to achieve our result we also present a new, simple sequential algorithm for constructing $(+6)$-spanners,
a result that could be of independent interest.

As a key ingredient, we provide a distributed construction of
\emph{multiple weighted BFS trees}.
Constructing a breadth-first search (BFS) tree is a central task in many computational settings. In the classic synchronous distributed setting, constructing a BFS tree from a given source is straightforward. Due to its importance, this task has received much attention in additional distributed settings, such as the asynchronous setting (see, e.g.,~\cite{Peleg_2000} and references therein).
Moreover, at the heart of many distributed applications lies a graph structure that represents the edges of \emph{multiple} BFS trees~\cite{HolzerW12,LenzenPP19}, which are rooted at the nodes of a given subset $S\subseteq V$, where $G=(V,E)$ is the underlying communication graph. Such a structure is used in
distance computation and estimation~%
\cite{HolzerW12,HolzerPRW14,LenzenPP19},
routing table construction~\cite{LenzenPP19},
spanner construction~\cite{Censor-HillelKP18,LenzenP13,LenzenPP19}, and more.

When the bandwidth is limited, constructing multiple BFS trees efficiently is a non-trivial task. Indeed, distributed constructions of multiple BFS trees in the \cgst{} model~\cite{Peleg_2000}, where in each round of communication every node can send $O(\log{n})$-bit messages to each of its neighbors, have been given in~\cite{HolzerW12,LenzenPP19}. They showed that it is possible to build BFS trees from a set of sources $S$ in $O(|S|+D)$ rounds, where $D$ is the diameter of the graph $G$;
it is easy to show that this is asymptotically tight.

In some cases, different edges of the graph may have different attributes, which can be represented using edge weights. The existence of edge weights has been extensively studied in various tasks,
such as finding or approximating lightest paths~%
\cite{Elkin17,Nanongkai14,HenzingerKN16,ElkinN16,LenzenPP19,HuangNS17,AgarwalRKP18,GhaffariL18,BernsteinN19,AgarwalR20},
finding a minimum spanning tree (MST) in the graph~\cite{Awerbuch87,GallagerHS83,SarmaHKKNPPW12},
finding a maximum matching~\cite{LotkerPP15,SarmaHKKNPPW12},
and more.
However, as far as we are aware, no study addresses the problem of constructing multiple weighted BFS (WBFS) trees, where the goal is not to find the lightest paths from the sources to the nodes, but rather the \emph{lightest shortest paths}. That is, the path in a WBFS tree from the source $s$ to a node $v$ is the lightest among all shortest paths from $s$ to $v$ in $G$.

Thus, we provide an algorithm that constructs multiple WBFS trees from a set of source nodes $S$ in the \cgst{} model. Our algorithm completes in $|S|+D-1$ rounds, which implies that no overhead is needed for incorporating the existence of weights.

While our multiple WBFS algorithm can be used for graphs that have initial edge weights, we actually use it to construct $(+6)$-spanners in unweighted graphs, to which we artificially add edge weights to distinguish ``desired'' edges from ``undesired'' ones.
More generally, this algorithm can be used to find \emph{consistent shortest paths} in an unweighted graph:
a set of shortest paths is consistent if, given four nodes $s,t,u$ and $v$ such that $u$ and $v$ lie on the shortest $(s,t)$-path, the shortest $(u,v)$-path is a subpath of the shortest $(s,t)$-path.
Such structures have many applications, but achieving them is non-trivial, especially in the distributed setting.
Our WBFS algorithm can be used for this goal, by assigning the edges of an unweighted graph random weights of polynomial size. 
The \emph{isolation lemma}~\cite{MulmuleyVV87} then guarantees that with high probability, the lightest shortest paths derived from this weight function will constitute a set of consistent shortest paths.

\subsection{Our contribution}
At a high level, our approach for building multiple WBFS trees
is to generalize the algorithm of Lenzen et al.~\cite{LenzenPP19}
in order to handle weights.
In~\cite{LenzenPP19}, the messages are pairs consisting of a source node and a distance,
which are prioritized by the distance traversed so far.
When incorporating weights into this framework it makes sense to use triplets instead of pairs,
where each triplet also contains the weight of the respective path. However, it may be that a node $v$ needs to send multiple messages
that correspond to the same source and the same distance but contain different weights,
since congestion over edges may cause the respective messages to arrive at $v$ in different rounds and,
in the worst case, in a decreasing order of weights.
The challenge in generalizing this framework therefore lies in guaranteeing that despite the need to consider weights,
we can carefully choose a total order to prioritize triplets,
such that not too many messages need to be sent,
allowing us to handle congestion.
Our construction and its proof appear in Section~\ref{sec: wbfs algorithm}, giving the following.

\newcommand{\ThmWBFS}{
Given a weighted graph $G=(V,E,w)$ and a set of nodes $S \subseteq V$,
there exists an algorithm for the \cgst{} model
that constructs a WBFS tree rooted at $s$, for every $s\in S$,
in $\size{S} +D-1$ rounds.
}
\begin{theorem}
\label{thm: weighted bfs}
	\ThmWBFS
\end{theorem}

Our multiple WBFS trees construction can have many applications, as it allows consistent tie-breaking. Here, we present one such application: pinning down the question of constructing
$(+6)$-spanners in the \cgst{} model.
The construction of additive spanners in the \cgst{} model was studied beforehand~\cite{LenzenP13,Censor-HillelKP18},
but the $+6$ case remained unresolved, for reasons we describe below.
Naturally, the quality of a spanner is measured by its sparsity,
which is the motivation for allowing some stretch in the distances to begin with,
and different spanners present different tradeoffs between stretch and sparsity.
The properties of our $(+6)$-spanner construction algorithm are summarized in the following theorem.\footnote{We use w.h.p.\ to indicate a probability that is at least $1-1/n^c$ for some constant $c\geq 1$ of choice.}

\newcommand{\ThmSpanner}
{
	There exists an algorithm for the \cgst{} model
	that constructs a $(+6)$-spanner with $O\left(n^{4/3} \log^{4/3} n\right)$ edges
	in $O\left(\frac{n^{2/3}}{\log^{1/3}n} +D\right)$ rounds
	and succeeds w.h.p.
}

\begin{theorem}
	\label{thm: 6ap}
	\ThmSpanner
\end{theorem}

Previous distributed algorithms for spanners similar to ours,
i.e., purely additive all-pairs spanners,
construct a $(+2)$-spanner with $\tilde O(n^{3/2})$ edges in $\tilde O(n^{1/2}+D)$ rounds~\cite{LenzenP13},
a $(+4)$-spanner with $\tilde O(n^{7/5})$ edges in $\tilde O(n^{3/5}+D)$ rounds~\cite{Censor-HillelKP18},
and a $(+8)$-spanner with $\tilde O(n^{15/11})$ edges in $\tilde O(n^{7/11}+D)$ rounds~\cite{Censor-HillelKP18}.
Hence, our algorithm is currently the best non-trivial spanner construction algorithm in terms of density, sparser even than the previous $(+8)$-spanner.
The parameters of our algorithm are tuned to achieve the best sparsity possible, and interestingly, one can change the algorithm and achieve a worst sparsity of $\tilde O(n^{15/11})$ edges in $\tilde O(n^{7/11}+D)$ rounds. These are the same parameters of the $(+8)$-spanner algorithm, but with a better stretch.
The option of getting even sparser spanners by allowing more stretch was essentially ruled out~\cite{AbboudB17},
while the question of
improving the running time remains open for all stretch parameters.

%
%
%

\subsection{Other spanner construction algorithms}
Previous distributed spanner construction algorithms
all build upon known sequential algorithms,
and present distributed implementations of them,
or of a slight variant of them~\cite{LenzenP13,Censor-HillelKP18}.
For example, many sequential algorithms start in a clustering phase,
where stars around high-degree nodes are added to the spanner one by one.
Implementing this directly in the distributed setting will take too long;
instead, we use a classical approach of choosing cluster centers at random,
which yields almost as good results, and can be implemented in a constant time.
Similar methods are used for implementing other parts of the construction.
However, the approach of finding a distributed implementation for a sequential algorithm fails for all known $(+6)$-spanner algorithms, as described next.
Thus, we introduce a new sequential algorithm for the problem,
and then present its distributed implementation.

There are three known approaches for the design of sequential $(+6)$-spanner algorithms.
The first, presented by Baswana et al.~\cite{BaswanaKMP10},
is based on measuring the quality of paths in terms of \emph{cost} and \emph{value},
and adding to the spanner only paths which are ``affordable''.
This approach was later extended by Kavitha~\cite{Kavitha17} to other families of additive spanners.
The second approach,
presented by Woodruff~\cite{Woodruff_2010},
uses a subroutine that finds almost-shortest paths between pairs of nodes,
and obtains a faster algorithm at the expense of a slightly worst sparsity guarantee.
The third approach, presented by Knudsen~\cite{Knudsen14},
is based on repeatedly going over pairs of nodes,
and adding a shortest path between a pair of nodes to the spanner
if their current distance in the spanner is too large.

Unfortunately, direct implementation in the \cgst{}
model of the known sequential algorithms is highly inefficient.
We are not aware of fast distributed algorithms that allow the computation
of the cost and value of paths needed for the algorithm of~\cite{BaswanaKMP10}.
Similarly, for~\cite{Woodruff_2010}, the almost-shortest paths subroutine seems too costly
for the \cgst{} model.
The algorithm of~\cite{Knudsen14} needs repeated updates of the
distances in the spanner between pairs of nodes
after every addition of a path to it,
which is a sequential process in essence,
and thus we do not find it suitable for an efficient distributed implementation.

A different approach for the distributed construction of $(+6)$-spanners
could be to adapt a distributed algorithm with different stretch guarantees
to construct a $(+6)$-spanner.
This approach does not seem to work:
the distributed algorithms for constructing $(+2)$-spanners~\cite{LenzenP13}
and $(+4)$-spanner~\cite{Censor-HillelKP18}
are both very much tailored for achieving the desired stretch,
and it is not clear how to change them in order to construct sparser spanners with  higher stretch.
The $(+8)$-spanner construction algorithm~\cite{Censor-HillelKP18}
starts with clustering, and then constructs a $(+4)$-\emph{pairwise spanner}
between the cluster centers.
Replacing the $(+4)$-pairwise spanner by a $(+2)$-pairwise spanner will indeed yield a $(+6)$-all-pairs spanner, as desired.
However, even using the sparsest $(+2)$-pairwise spanners~\cite{Censor-HillelKP18,AbboudB16-1},
the resulting $(+6)$-spanner may have $\tilde{O}(n^{5/3})$ edges,
denser than our new $(+6)$-spanner
and than the known $(+8)$-spanner~\cite{Censor-HillelKP18}.

Thus, we start by presenting a new sequential algorithm for the construction of $(+6)$-spanners,
an algorithm that is more suitable for a distributed implementation,
and then discuss its distributed implementation.
Our construction starts with a clustering phase,
and then adds paths that minimize the number of additional edges that need to be added to the spanner.
To implement our construction in the \cgst{} model,
we assign weights to the edges and use our WBFS algorithm
to find shortest paths with as few edges as possible that are not yet in the spanner.
Note that although the graph and the spanner we construct for it are both unweighted,
the ability of our multiple WBFS algorithm to handle weights is crucial for our solution.

A $(+6)$-spanner must contain $n^{4/3}/2^{O(\sqrt{\log n})}$ edges~\cite{AbboudB17}.
The best sequential algorithms~\cite{Knudsen14,BaswanaKMP10} construct a spanner
with $O(n^{4/3})$ edges.
Our distributed algorithm constructs a spanner with $O(n^{4/3}\log^{4/3}n)$ edges,
which is slightly denser than optimal
but still sparser than the $O(n^{4/3}\log^{3}n)$ edges in the fast sequential construction of~\cite{Woodruff_2010}.

\subsection{Related work}

A natural approach for building multiple (unweighted) BFS trees in the \cgst{} model is to start the construction simultaneously from all sources, sending pairs composed of a source id and the distance the message have traversed so far. The main challenge in this approach appears when multiple such pairs should be sent on one edge in a single round. In this case, the algorithm needs to locally prioritize the message, in a way that will not compromise the correctness of the constructed BFS trees.

One approach for prioritizing the messages is by source id, i.e., have a queue at each node sorted by source id. However, this alone does not work and might create incorrect BFS trees.
Attempts to overcome the problems of this approach by having a different queue for each outgoing edge~\cite{HolzerW12}
are also insufficient~\cite{HWpersonal}.

An elegant way to resolve this is to sort messages not by source id, but by the distance they traversed so far, with a preference to messages that traversed smaller distances.
Such an algorithm was suggested by Lenzen and Peleg~\cite{LenzenPP19},
finishing in $(\size{S} +D-1)$ rounds when executed form a set $S$ of sources.
We base our multiple WBFS construction algorithm on this algorithm.
A variant of the Lenzen-Peleg algorithm, with a reduced message complexity, was recently suggested in~\cite{PontecorviR18}.

%


Spanners were first introduced in 1989~\cite{PS89, PelegU89a}, and since then have been a topic for wide research due to their abundant applications. Prime examples for the need for sparse spanners can be found in synchronizing distributed networks~\cite{PelegU89a}, information dissemination~\cite{Censor-HillelHKM12}, compact routing schemes~\cite{Chechik13a, PelegU89b,ThorupZ01}, and more.

Distributed constructions of various spanners have been widely studied~\cite{LenzenP13,LenzenPP19,Pettie10,Censor-HillelKP18,
BaswanaKMP10, BaswanaS07, DerbelG08, DerbelGP07, DerbelGPV08, DubhashiMPRS05, E05, ElkinZ06, DerbelGPV09,ParterY18,GhaffariK18,Parter17,GrossmanP17,Censor-HillelD18,ElkinMatar19,ElkinN19,BarenboimEG18}. Lower bounds were given in~\cite{Pettie10,Censor-HillelKP18,AbboudCK16}.
However, obtaining an efficient and sparse ($+6$)-all-pairs spanner has remained an open question~\cite{Censor-HillelKP18}.

Several lower bounds for the time complexity of spanner construction in the \cgst{} model were presented in~\cite{Censor-HillelKP18},
but these are applicable only to pairwise spanners with a bounded number of pairs, and not to all-pairs spanners.
A lower bound from~\cite{Pettie10} states that the construction of a spanner with $\tilde O(n^{4/3})$ edges,
such as the one we build,
must take $\tilde\Omega(n^{3/8})$ rounds.
This lower bound does not take into account the bandwidth restrictions at all (it is proven for the \lcl{} model),
and so we believe that a higher lower bound for the \cgst{} model should apply,
but this is left as an intriguing open question.

\section{Preliminaries}
\label{sec:preliminaries}
All graphs in this work are simple, connected and undirected.
A graph can be unweighted, $G=(V,E)$,
or weighted $G=(V,E,w)$ with $w:E\to \set{0,\ldots, W}$,
in which case we assume $W\in \text{poly}(n)$.
Given a path $\rho$ in a weighted graph $G$,
we use $\size{\rho}$ to denote the \emph{length} of $\rho$,
which is the number of edges in it,
and $w(\rho)$ to denote the \emph{weight} of the path,
which is the sum of its edge-weights.
The \emph{distance} between two nodes $u,v$ in a graph $G$,
denoted $\delta_G(u,v)$,
is the minimum length of a path in $G$ connecting $u$ and $v$.
The \emph{diameter} of a graph (weighted or unweighted)
is $D=\max_{u,v\in V}\set{\delta_G(u,v)}$.

We consider the \cgst{} model of computation\cite{Peleg_2000},
where the nodes of a graph communicate synchronously by exchanging $O(\log n)$-bit
messages along the edges.
The goal is to distributively solve a problem
while minimizing the number of communication rounds.

\subparagraph{WBFS trees:}
We are interested in a weighted BFS tree, which consists of all \emph{lightest shortest paths} from the root, formally defined as follows.

\begin{definition}
	Given a connected, weighted graph $G=(V,E,w)$
	and a node $s\in V$,
	a \emph{weighted BFS tree (WBFS)} for $G$ rooted at $s$
	is a spanning tree $T_s$ of $G$
	satisfying the following properties:
	\begin{enumerate}[(i)]
	  \item For each $v\in V$, the path from $s$ to $v$ in $T$ is
	    a shortest path in $G$ between $s$ and $v$.
	  \item For each $v\in V$, no shortest path from $s$ to $v$ in $G$
	  is lighter than the path from $s$ to $v$ in~$T$.
	\end{enumerate}
\end{definition}

We emphasize that this is different than requiring a subgraph containing all \emph{lightest paths} from the root.
One may wonder if a WBFS tree always exists,
but this is easily evident by the following refinement of a (sequential) BFS search, returning a WBFS tree:
go over the nodes in an order of non-decreasing distances from the source $s$, starting with $w(s)=0$;
each node $v$ chooses as a parent a neighbor $u$ that was already processed and minimizes $w(v)=w(u)+w(u,v)$, and adds the edge $\{u,v\}$ to the tree.
Each node has a single parent, so this is indeed a tree; the node ordering guarantees that this is indeed a BFS tree, assuring~(i); and the parent choice guarantees the paths are lightest among the shortest, assuring~(ii).
In fact, this can be seen as an algorithm implementing a consistent tie-breaking strategy between paths of equal lengths, and thus, it is not surprising its output is a tree.

\subparagraph{Spanners:}
Given a graph $G=(V,E)$, a subgraph $H=(V,E')$ of $G$ is called an \emph{$(\alpha,\beta)$-spanner} if for every $u,v \in V$ it holds that
$\delta_H(u,v) \leq \alpha\delta_G(u,v) + \beta$.
The parameters $\alpha$ and $\beta$ are called the \emph{stretch parameters}.

When $\alpha=1$, such a spanner is called a \emph{purely additive spanner}. In this paper we focus on purely additive ($+6$)-spanners, i.e., $\alpha=1$ and $\beta=6$.

For completeness, we mention that when $\beta=0$, such a spanner is called a \emph{multiplicative spanner}. In addition, while sometimes the stretch parameters need to be guaranteed only for some subset of all the pairs of nodes of the graph (such as in \emph{pairwise spanners}), we emphasize that our construction provides the promise of a $+6$ stretch for \emph{all} pairs.

\section{Multiple Weighted BFS Trees}
\label{sec: wbfs algorithm}
In the \cgst{} model, the problem of finding a WBFS tree requires each node
to know its parent in the WBFS tree,
and the unweighted and weighted distances to the source within the tree.
This allows the node to send messages to the source node through the lightest among all shortest paths.
When there are multiple sources, each node should know the parent leading to each of the sources in $S$.

We define data structures for representing multiple WBFS trees.
Given a node $v \in V$, the \textit{$S$-proximity-list} (or \textit{proximity list} for short) of $v$, noted $\proxlv$,
is an ascending lexicographically ordered list of triples $(d(s,v),s,w(s,v))$,
where $d(s,v)$ and $w(s,v)$ are the length and weight of the path from $s$ to $v$ in $T_s$.
Two different triples are ordered such that $(d(s,v),s,w(s,v)) < (d(t,v),s,w(t,v))$
if $d(s,v) < d(t,v)$, or $d(s,v) = d(t,v)$ and $s < t$,
where $s$ and $t$ may be compared by any predefined order on the node identifiers.
Note that $T_s$ contains a single path from $s$ to $v$,
so $\proxlv$ cannot contain two triplets with $d(s,v) = d(t,v), s = t$ and $w(s,v)\neq w(t,v)$.

The \textit{$S$-path-map} (or \textit{path-map} for short) of $v$
is a mapping from each source $s \in S$ to the parent of $v$ in $T_s$,
noted by $\pathmv$.
The list $\pathmv$ is sorted with respect to the order of $\proxlv$,
such that the first records of $\pathmv$ belong to sources closest to $v$.

Algorithm~\ref{alg: parallel weighted bfs},
which constructs multiple WBFS trees from a set $S$ in the \cgst{} model,
is based on carefully extending the distributed Bellman-Ford-based algorithm of Lenzen et al.~\cite{LenzenPP19}.
The heart of the algorithm is a loop (Line~\ref{line: for loop}),
and each iteration of it takes a single round in the \cgst{} model.
We show that $\size{S} +D-1$ iterations of the loop suffice in order to construct the desired WBFS trees.

The algorithm builds the WBFS trees by gradually updating the proximity list
and the path map of each node.
Each round is composed of two phases:
updating the neighbors about changes in the proximity list,
and receiving updates from other nodes.
The path map is only used by the current node,
and therefore changes to it are not sent.

Ideally, each node would update its neighbors
regarding all the changes made to its proximity list.
However, due to bandwidth restrictions,
a node cannot send the entire list in each round.
Therefore, at each round each node sends to all of its neighbors the
lexicographically smallest triplet in its proximity list that it has not yet
sent,
while maintaining a record noting which triplets have been sent and which are waiting.
Each triplet is only sent once, though a node may send multiple triplets
regarding a single source.

A node uses the messages received in the current round in order to update its
proximity list and path map for the next round.
A triplet $(d_s, s, w_s)$ received by a node $v$ from a neighbor $u$
represents the length $d_s$ and weight $w_s$ of some path $\rho$ from
$s$ to $u$ in the graph. The node $v$ then considers the extended path
$\rho' = \rho \circ v$ from $s$ to $v$,
compares it to its currently known best path from $s$ to $v$,
and updates the proximity list and path map in
case a shorter path has been found, or a lighter path with the same length.

\RestyleAlgo{boxruled}
\LinesNumbered
\begin{algorithm}[t]
	\DontPrintSemicolon
	\caption{Weighted distributed Bellman-Ford algorithm
            for node $v$}
	$L_v \gets ()$ \\
	\label{alg: parallel weighted bfs} 
    \For{$s \in S$}{
    	$\algpath_v(s) \gets \bot $
    }
    \If{$v\in S$}{
        $\algprox_v\gets((0,v,0))$\\
        $\sent_v(0,v,0) \gets \false$ \tcc*{A variable marking sent triplets}
    }
    \For{$\size{S}+D-1$ \text{rounds}}{
    	\label{line: for loop}
        \If{$\exists (d_s,s,w_s) \in \algprox_v \mbox{ such that }
            \sent_v(d_s,s,w_s) = \false$}{
            $(d_s,s,w_s)\gets \min
            \set{(d_{t},t,w_{t}) \in \algprox_v \mbox{ such that }
                \sent_v(d_{t},t,w_{t}) = \false}$\\
            send $(d_s,s,w_s)$ to all neighbors\\
            $\sent_v(d_s,s,w_s)\gets\true$
        }
        \For{received $(d_s,s,w_s)$ from $u \in V$}{
            $d_s \gets d_s + 1$\\
            $w_s \gets w_s + w(u,v)$\\
            \If{$\nexists (d'_s,s,w'_s) \in \algprox_v \mbox{ such that }
                (d'_s < d_s \text{ or }
                (d'_s = d_s \text{ and } w'_s < w_s))$}{
                $\algprox_v \gets \algprox_v \setminus \set{(\cdot, s, \cdot)}$\\
                $\algprox_v \gets \algprox_v \cup \set{(d_s, s, w_s)}$\\
                $\algpath_v(s) \gets u$ \\
                $\sent_v(d_s, s, w_s)\gets\false$
            }
        }
    }
\end{algorithm}

To prove correctness,
we generalize the proof of~\cite{LenzenPP19} to handle weights,
and show that our algorithm solves the
\emph{weighted $(S,d,k)$-detection} problem:
each node should learn which are the sources from $S$ closest to it,
but at most $k$ of them and only up to distance $d$.
This is formally defined as follows.

\begin{definition}
\label{def:weighted sdk detection}
Given a weighted graph $G=(V,E,w)$,
a subset $S \subseteq V$ of source nodes,
and a node $v \in V$,
let $\proxlv$ denote the S-proximity-list
and let $\pathmv$ denote the path map of the node $v$.
The \emph{weighted $(S,d,k)$-detection problem} requires
that each node $v \in V$ learns the first $\min\set{k,\lambda_v^d}$ entries of $\proxlv$ and $\pathmv$,
where $\lambda_v^d$ is the number of sources $s \in S$ such that $d(s,v) \leq d$.
\end{definition}

Given a node $v$,
$\algprox_v$ is a variable in Algorithm~\ref{alg: parallel weighted bfs}
holding the proximity list of $v$,
and we denote by $\algprox_v^{(r)}$ the state of the list $\algprox_v$
at the beginning of round $r$ of the algorithm,
and by $\algprox_v^{(\infty)}$ the value of $\algprox_v$ at the end of the algorithm.
Recall that $\proxlv$ is the true proximity list,
so our goal is proving $\algprox_v^{(\infty)} = \proxlv$,
i.e., proving that the algorithm obtains the correct values of the proximity list.

We use similar notations for the path map $\pathmv$.
Since the records of $\algpath_v$ are updated under the same conditions as the records of $\algprox_v$,
the correctness of $\algpath_v$ at the end of the algorithm with respect to $\pathmv$ immediately follows,
and we omit the details.

We start by showing that if there was no bound on the number of rounds,
then the values of $\algprox_v$ would have eventually converged to the true
values of $\proxlv$.

\newcommand{\LemmaEventualDistances}
{
Given a graph $G=(V,E,w)$ and a set $S\subseteq V$,
if we let the for loop in Line~\ref{line: for loop}
of Algorithm~\ref{alg: parallel weighted bfs}
to run forever,
then there exists a round $r_0\in \bbN$ such that
no node $v\in V$ sends messages or modifies $\algprox_v$ after round $r_0$.
Moreover, $\algprox_v^{(r_0)} = \proxlv$,
i.e., for every $(d_s,s,w_s)\in \algprox_v^{(r_0)}$, it holds that
$d_s = d(v,s)$ and
$w_s = \min \set{w(\rho) \mid
\rho\text{ connects } v \text{ with } s \text{, and } \size{\rho} = d_s}$.
}

\begin{lemma}
\label{lem: eventual distances are correct}
\LemmaEventualDistances
\end{lemma}

\begin{proof}
In each iteration of the algorithm, each node $v \in V$ sends to all of its
neighbors the first triplet $(d_s,s,w_s)$ such that $\sent_v(d_s,s,w_s) = \false$. Each triplet received
is sent at most once. Therefore, if we show the existence of a round $r_0$ where
for each $v \in V$, all messages in $\algprox_v^{(r_0)}$ have been sent in previous
rounds, it implies that no message is sent in round $r_0$, and hence
$\algprox_v^{(r_0 + 1)} = \algprox_v^{(r_0)}$. This claim is applied inductively,
concluding that for any round $r > r_0$, it holds that $\algprox_v^{(r)} =
\algprox_v^{(r_0)}$.
We prove the existence of the round $r_0$ by showing that the number of messages that are sent by each node $v \in V$ is finite.

Each triple $(d_s, s, w_s)$ is inserted into $\algprox_v$
when a message is received from another node $u \in V$,
and such a message implies that
a path from $s$ to $v$ through $u$ of length $d_s$ and weight $w_s$ exists.
Thus, the number of messages $(d_s, s, w_s)$ sent by a node $v$ is upper bounded by the number of paths from $s$ to $v$ of length $d_s$.
Furthermore, the algorithm does not insert a message into $\algprox_v$
if it has already inserted a lexicographically smaller message from the same
source.
As the graph is finite, the number of paths is bounded,
and eventually no node adds further triplets to its lists
or sends additional messages.

It remains to show that for all $s \in S$ and $v \in V$ it holds that $(d(s,v),
s, w(s,v)) \in \algprox_v^{(r_0)}$.

First, we show that if a triplet $(d_s, s, w_s)$ is added to $\algprox_v$ in some round $r$, then there exists a path $\rho$ from $s$ to $v$ such that $d_s = |\rho|$ and $w_s = w(\rho)$.
For a source $s \in S$, we insert the triplet $(0, s, 0)$ into $\algprox_s$ at the beginning of the algorithm,
so the claim is true at initialization.
Assume there exists a round where a triplet $(d_s,s,w_s)$ is inserted into $\algprox_v$ but no corresponding path exists,
and let $r$ be the first such round.
This implies that there exists a node $u \in V$ that is a neighbor
of $v$, which sends the message $(d_s - 1, s, w_s - w(u,v))$ to $v$ in round $r$.
The triplet $(d_s - 1, s, w_s - w(u,v))$ must have been inserted into $\algprox_u$
in some round $r'<r$,
and by the minimality of $r$ there exists a path $\rho$ from $s$ to $u$ where $|\rho| = d_s - 1$
and $w(\rho) = w_s - w(u,v)$.
Since $u$ is a neighbor of $v$, the path $\rho' =
\rho \circ {v}$ is valid, satisfying $|\rho'|=d_s$ and $w(\rho')=w_s$,
contradicting the assumption.

To complete the proof,
we claim that the correct triplet $(d(s,v), s, w(s,v))$ is indeed
added to $\algprox_v$ at some round of the algorithm,
and is not removed.
Consider the path from $s$ to $v$ in the WBFS tree $T_s$,
denoted by $\rho = (v_0 = s, v_1, \ldots, v_{d(s,v)}=v)$.
At initialization, the triplet $(0,s,0)$ is added to $\algprox_s$.
From then, at each round there exists some $i \leq d(s,v)$
such that $(d(s,v_i), s, w(s,v_i)) \in \algprox_{v_i}$
and $\sent_{v_i}(d(s,v_i), s, w(s,v_i))=\false$.
Since we proved that this message is eventually sent,
this implies that in the beginning of the next round,
the triplet $(d(s,v_{i+1}), s, w(s,v_{i+1}))$
is added to  $\algprox_{v_{i+1}}$.
By the definition of a WBFS tree, all other paths from $s$ to $v$ must be
longer or not lighter, implying the triplet cannot be discarded for a
lexicographically smaller triplet.
This concludes that for any source $s$ and node $v$,
it holds that $(d(s,v), s, w(s,v)) \in \algprox_v^{(r_0)}$.
\end{proof}

Lemma~\ref{lem: eventual distances are correct}
shows that without the limit on the number of rounds,
the algorithm would compute the right values;
however, it does not bound the number of rounds needed for this to occur.
Next, we show that $\size{S}+D-1$ rounds suffice.
We cannot apply the claims of~\cite{LenzenPP19} directly,
since the existence of weights restricts
the number of viable solutions even further,
causing more updates to the proximity list
and an increase in the number of messages sent.
However,
we do use a similar technique:
we bound the number of rounds in which
the $k$ smallest entries of $\algprox_v$ can change.

For an entry $(d_s,s,w_s) \in \algprox_v^{(r)}$,
let $\ell_v^{(r)}(d_s, s, w_s)$
denote the index of the entry in the
lexicographically ordered list $\algprox_v^{(r)}$
at the beginning of round $r$.
For completeness,
we define $\ell_v^{(r)}(d_s, s, w_s)=-\infty$
if $(d_s, s, w_s)$ did not appear in $\algprox_v$ at the beginning of round $r$,
and $\algprox_v=\infty$ if the triplet was removed from $\algprox_v$ before the beginning of this round.
Note that a removed triplet is never returned to the list,
since the lexicographical order is transitive.

\begin{lemma} 
\label{helper lemma}
For a triplet $(d_s, s,w_s)$, the following holds:
\begin{enumerate}[(i)]
  \item $\ell_v^{(r)}(d_s, s,w_s)$ is non-decreasing with r.
  \item
  	When a triplet $(d_s, s,w_s)$ is sent
  	from a node $u$ to a node $v$ at round $r$,
  	it causes the addition of a new triplet $(d'_s, s,w'_s)$ to $\algprox_v$ at the end of round $r$.
  	This triplet satisfies $d'_s = d_s + 1$, $w'_s = w_s + w(u,v)$,
  	and $\ell_u^{(r)}(d_s, s,w_s) \leq \ell_v^{(r+1)}(d'_s, s,w'_s)$.
\end{enumerate}
\end{lemma}

Part~(i) follows from the fact that the number of triplets below $(d_s,s,w_s)$ cannot decrease.
To prove part~(ii), we show that all the triplets below $(d_s, s,w_s)$
in $\algprox_u$ are sent from $u$ to $v$ and added to $\algprox_v$
before $(d_s, s,w_s)$ is sent and added.

\begin{proof}
Part (i) is a consequence of the method used by our algorithm for managing the list $\algprox_v$.
According to our algorithm, triplets are not removed from
$\algprox_v$ when they are sent.
The only case in which a triplet $(d_t,t,w_t)$ is removed from $\algprox_v$ is when a lexicographically smaller triplet $(d'_t,t,w'_t)$ is added to the list instead.
When this happens in round $r$,
it holds that
$\ell_v^{(r)}(d_t,t,w_t) \geq \ell_v^{(r+1)}(d'_t,t,w'_t)$,
since the new triplet is lexicographically smaller.
Hence, for every other triplet $(d_s,s,w_s)\in \algprox_v$,
the number of lexicographically smaller triplets in $\algprox_v$ cannot decrease throughout the algorithm.

We now turn to prove part (ii) of the lemma.
	Note that the only claim which does not follows trivially from the algorithm is the inequality  $\ell_u^{(r)}(d_s, s,w_s) \leq \ell_v^{(r+1)}(d'_s, s,w'_s)$.
By the fact that the triplet $(d_s,s,w_s)$ is sent by the node $u$ in round $r$,
we conclude that
the $\ell_u^{(r)} (d_s,s,w_s)-1$ triplets preceding it in the list $\algprox_u^{(r)}$
have already been sent by $u$ in earlier rounds,
and arrived at the node $v$.
For each such triplet $(d_t,t,w_t)$,
either $d_t \leq d_s$, or $t < s$ and $d_t = d_s$.
Therefore, when added to $\algprox_v$ as
$(d_t+1,t,w_t+w(u,v))$ it is lexicographically smaller than $(d'_s,s,w'_s)$.
At round $r$,
either $(d_t+1,t,w_t+w(u,v))$ is in $\algprox_v^{(r)}$
or it was replaced by a lexicographically smaller triplet
containing $t$.
Thus, there are at least $\ell_u^{(r)} (d_s,s,w_s)-1$ triplets
smaller than $(d'_s,s,w'_s)$ in $\algprox_v^{(r+1)}$,
and hence $\ell_u^{(r)} (d_s,s,w_s) \leq \ell_v^{(r+1)}(d'_s,s,w'_s)$.
\end{proof}

Lemma \ref{helper lemma} implies that as the algorithm progresses, messages at higher
indexes of the proximity list are sent and updated. This can be used to obtain
an upper bound on the round in which a triplet at a certain index of the
proximity list can be sent or received,
as formalized by the next lemma.

\begin{lemma}
\label{lem: first entries stabilize}
In round $r\in \bbN$ of Algorithm~\ref{alg: parallel weighted bfs},
a node $v\in V$ can:
\begin{enumerate}[(i)]
  \item send a message $(d_s,s,w_s)$ only if
      \[
      d_s + \ell_v^{(r)}(d_s,s,w_s) \geq r
      \]
  \item add to $\algprox_v$ a triplet $(d_s,s,w_s)$ only if
      \[
      d_s + \ell_v^{(r+1)}(d_s,s,w_s) > r
      \]
\end{enumerate}
\end{lemma}

Part~(i), when put in words, is rather intuitive:
while a triplet might need to wait before being sent,
the waiting time is bounded from above by the distance the triplet
has traversed from its source,
plus the number of triplets that were to be sent before it.
Part~(ii) is complementary to part~(i):
the time before a triplet is added, is, once more,
bounded by the distance it traversed plus the number of lexicographically smaller triplets.

\begin{proof}
We start by showing that,
for a given round $r$,
if Lemma~\ref{lem: first entries stabilize}(i) holds for all nodes
then Lemma~\ref{lem: first entries stabilize}(ii) holds as well.
Consider a triplet $(d'_s,s,w'_s)$ that is added to
$\algprox_v$ as a result of a message $(d_s,s,w_s)$
sent from $u$ to $v$ in round $r$,
where $d'_s=d_s+1$ and $w'_s=w_s+w(u,v)$.
Lemma~\ref{lem: first entries stabilize}(i) implies that $d_s + \ell_u^{(r)}(d_s,s,w_s) \geq r$,
and by Lemma~\ref{helper lemma}(ii) we have that
$\ell_v^{(r+1)}(d'_s,s,w'_s) \geq \ell_u^{(r)}(d_s,s,w_s)$.
As $d'_s > d_s$, we conclude
\[d'_s + \ell_v^{(r+1)}(d'_s,s,w'_s) > d_s + \ell_u^{(r)}(d_s,s,w_s)
\geq r,\]
which implies Lemma~\ref{lem: first entries stabilize}(ii).

Next, we prove by induction that both parts of the lemma hold.
In round $1$, Lemma~\ref{lem: first entries stabilize}(i) holds trivially, since by
definition $\ell_v^{(1)}(d_s,s,w_s) \geq 1$.
Assume that Lemma~\ref{lem: first entries stabilize}
holds at round $r-1$;
we show the lemma holds at round $r$.
Since Lemma~\ref{lem: first entries stabilize}(i) implies Lemma~\ref{lem: first entries stabilize}(ii),
it is sufficient to show that every message $(d_s,s,w_s)$
sent by some node $v \in V$ in round $r$
satisfies $d_s + \ell_v^{(r)}(d_s,s,w_s) \geq r$.

Observe that if $(d_s,s,w_s)$ is sent by a node $v$ in round $r$,
then the triplet must have been added to $\algprox_v$ in some round $r' \leq r-1$.
If $r' = r-1$, according to the induction hypothesis, Lemma~\ref{lem: first entries stabilize}(ii) holds and
$d_s + \ell_v^{(r)}(d_s, s, w_s) > r-1$,
implying $d_s + \ell_v^{(r)}(d_s, s, w_s) \geq r$,
since all the terms are integers.

Otherwise $r' < r-1$. In this case, in round $r-1$ the triplet $(d_s,s,w_s)$
appeared in $\algprox_v$ and was not yet sent. Since $(d_s,s, w_s)$ is sent in round
$r$, a different triplet $(d_{t}, t,w_{t})$ with $t \neq s$ must have
been sent in round $r-1$, implying:
\[
d_s + \ell_v^{(r-1)}(d_s, s, w_s) > d_{t} + \ell_v^{(r-1)}(d_{t}, t, w_{t}).
\]

By Lemma~\ref{helper lemma}(i), we have that $\ell_v^{(r)}(d_s,s,w_s) \geq
\ell_v^{(r-1)}(d_s,s,w_s)$, and combined with the induction hypothesis for Lemma~\ref{lem: first entries stabilize}(i) in
round $r-1$ we conclude:
\[
d_s + \ell_v^{(r)}(d_s,s,w_s)  \geq d_s + \ell_v^{(r-1)}(d_s,s,w_s) > d_{t} + \ell_v^{(r-1)}(d_{t}, t, w_{t}) \geq r-1.
\]
This gives that $d_s + \ell_v^{(r)}(d_s,s,w_s) \geq r$, since all the terms
are integers.
\end{proof}

Lemma~\ref{lem: eventual distances are correct} implies that eventually,
the lists $\algprox_v$ converge to contain the correct values,
and Lemma \ref{lem: first entries stabilize} restricts the number of rounds in which specific list entries may change.
From this, we conclude that the algorithm solves the weighted $(S,d,k)$-detection problem.

\begin{lemma}
\label{lem: first r entries are correct}
Given an instance of the weighted $(S,d,k)$-detection problem,
for every $v\in V$
and round $r$ of an execution
of Algorithm~\ref{alg: parallel weighted bfs} with
\[
r\geq \min\set{d,D} + \min\set{k,\size{S}},
\]
the truncation of $\algprox_v^{(r)}$
to the first $\min\set{k,\lambda^d_v}$ entries,
where $\lambda^d_v$ is the number sources $s\in S$ such that $d(s,v)\leq d$,
solves weighted $(S,d,k)$-detection problem.
\end{lemma}

This lemma says that the truncated list is
correct at the beginning of the relevant round.
To prove it, we use
Lemma~\ref{lem: first entries stabilize}(ii) to show that the values in the truncated list cannot change at round $r$ or later,
and Lemma~\ref{lem: eventual distances are correct} to deduce they are correct.

\begin{proof}
Assume w.l.o.g that $d\leq D$,
as $D$ bounds the distance to any source,
and $k\leq \size{S}$,
as otherwise $v$ needs to learn about all sources.

By Lemma~\ref{lem: eventual distances are correct},
there is a round $r_0$ when all entries of $\algprox_v^{(r_0)}$ are correct,
and let $(d_s,s,w_s)$ be a triplet in one of the first $\min\set{k,\lambda^d_v}$ entries of $\algprox_v^{(r_0)}$.
Since $(d_s,s,w_s)$ is one of the first $\lambda^d_v$ entries and  $\algprox_v^{(r_0)} = \proxlv$,
we have $d_s\leq d$.

Let $r$ be the round when $(d_s,s,w_s)$ is inserted to the list $\algprox_v$.
By Lemma~\ref{lem: first entries stabilize}(ii),
$r< d_s + \ell_v^{(r+1)}(d_s,s,w_s)$.
By Lemma~\ref{helper lemma}(i),
when the triplet is inserted to the list,
it is already placed in one of the first $\min\set{k,\lambda^d_v}$ entries,
i.e.,  $\ell_v^{(r+1)}(d_s,s,w_s) \leq \min\set{k,\lambda^d_v}\leq k$.
Hence,
\[
r < d_s + \ell_v^{(r+1)}(d_s,s,w_s)\leq d+k.
\]
Since this claim holds for any of the first $\min\set{k,\lambda^d_v}$
entries,
these were all correct at the beginning of round $d+k$,
and in all the succeeding rounds.
\end{proof}

The construction of multiple WBFS trees
is an instance of the $(S,D,\size{S})$-detection problem.
Lemma~\ref{lem: first r entries are correct} shows that
after $\size{S}+D-1$ rounds of Algorithm~\ref{alg: parallel weighted bfs}
on such an instance,
all the entries of the list $\algprox_v^{(\size{S}+D)}$ are correct,
yielding the main result of this section.

\begin{theorem-repeat}{thm: weighted bfs}
\ThmWBFS
\end{theorem-repeat}

\section{A ($+6$)-Spanner Construction}
\label{sec: spanner construction}
In this section we discuss the distributed construction of
$(+6)$-spanners.
First, we present a template for constructing
a $(+6)$-spanner
and analyze the stretch and sparsity of the constructed spanner.
Then,
we provide an implementation of our template
in the \cgst{} model and analyze its running time.

A \emph{cluster} $C_i$ around a \emph{cluster center} $c_i\in V$ is a subset of the set of neighbors of $c_i$ in $G$.
A node belonging to a cluster is \emph{clustered}, while the other nodes are \emph{unclustered}. 

Our algorithm starts by randomly choosing cluster centers, and adding edges between them to their neighbors, where each neighbor arbitrarily chooses a single center to connect to.
Then, additional edges are added, to connect each unclustered node to all its neighbors.
Next, shortest paths between clusters are added to the spanner.
In order to find these shortest paths in the \cgst{} model,
we use the WBFS construction algorithm
to build WBFS trees from random sources.
At the heart of our algorithm stands the path-hitting framework of Woodruff~\cite{Woodruff_2010}:
a shortest path in the graph which has many edges between clustered nodes, must go through many clusters.
This fact is used in order to show that
a path with many missing edges (edges not in $H$)
is more likely to have an adjacent source of a WBFS tree,
and thus it is well approximated by a path within the spanner.

Woodruff's algorithm starts with a similar clustering step.
However, in order to add paths between clusters,
it uses an involved subroutine
that finds light almost-shortest paths between pairs of nodes.
This subroutine seems too global to be implemented efficiently in a distributed setting,
so in our construction
it is replaced by only considering lightest shortest paths,
which we do using the WBFS trees defined earlier.

Our algorithm constructs a $(+6)$-spanner with
$O(n^{4/3}\log^{4/3}n)$ edges in $\tilde{O}(n^{2/3}+D)$ rounds,
as stated next.

\begin{theorem-repeat}{thm: 6ap}
\ThmSpanner
\end{theorem-repeat}

Lemmas~\ref{lemma: 6ap size} and \ref{lemma: 6ap stretch}
analyze the size and stretch of \sixap{} given below.
The number of rounds of its distributed implementation
is analyzed in Lemma~\ref{lemma: 6ap complexity}, giving Theorem~\ref{thm: 6ap}.
We use $c>2$ to denote a constant that can be chosen according to the desired exponent of $1/n$ in the failure probability.

\subsubsection*{\sixap{}}
Input: a graph $G=(V,E)$, a constant $c>2$;\\
Output: a subgraph $H$ of $G$;\\
Initialization:
$n\gets |V|$; $H\gets \left(V,\emptyset\right)$; $k\gets 1$

\paragraph{Clustering.}
Pick each node as a \emph{cluster center}
w.p.\ $\frac{c}{n^{1/3}\log^{1/3}n}$,
and denote the set of selected nodes by
$\calC = \set{ c_1, c_2, \ldots }$.
For each $c_i$, initialize a cluster $C_i\gets\emptyset$.

For each node $v\in V$,
choose a neighbor $c_i$ of $v$ which is a cluster center,
if such a neighbor exists,
add the edge $(v,c_i)$ to $H$, and add $v$ to $C_i$.
If none of the neighbors of $v$ is a cluster center,
add to $H$ all the edges adjacent to $v$. Let $H_0\gets H$.

\paragraph{Path Buying.}$ $\\
While $k\leq \frac{8cn^{2/3}}{\log^{1/3} n}$ do:
\begin{enumerate}
\item
    $S_k \gets \emptyset$
\item
    Add each cluster center $c_i\in \calC$ to $S_k$
    w.p.\ $\frac{8c^2\log n}{k}$,
    independently of the other centers
\item
    For each pair $(c_i,c_j)\in \calC \times S_k$:
    \begin{enumerate}
    \item
        $A \gets \emptyset$ \hfill /* $A$ is a set of paths */
    \item
        For each $v\in C_j$:
        \begin{enumerate}
        \item
            Among all the shortest paths from $c_i$ to $v$,
            let $P_v$ be a path with minimum $\size{P_v \setminus H_0}$
        \item
            If $\size{P_v \setminus H_0} <2k$, add $P_v$ to $A$
        \end{enumerate}
    \item
        If $A\neq \emptyset$,
        add to $H$ one of the shortest among the paths of $A$
    \end{enumerate}
\item
    $k\gets 2k$
\end{enumerate}



\begin{lemma}
\label{lemma: 6ap size}
\sixap{} outputs a subgraph $H$ of $G$
with $O(n^{4/3} \log^{4/3} n)$ edges,
with probability at least~$1-O(n^{-c+1})$.
\end{lemma}

\begin{proof}
The algorithm starts with $H=(V,\emptyset)$ and only adds edges from $G$,
so $H$ is indeed a subgraph of $G$ over the same node set.

In the first part of the clustering phase,
each node adds to $H$ at most one edge,
connecting it to a single cluster center,
for a total of $O(n)$ edges.
Then, the probability that a node of degree at least $n^{1/3}\log^{4/3} n$
is left unclustered
is at most $\left(1-\frac{c}{n^{1/3}\log^{1/3}n}\right)^{n^{1/3}\log^{4/3}n}$,
which is $O(n^{-c})$.
A union bound implies that all nodes of degree at least $n^{1/3}\log^{4/3}n$
are clustered w.p.~$1-O(n^{-c+1})$,
and thus the total number of edges added to $H$ by unclustered nodes
in the second part of the clustering phase
is $O(n^{4/3}\log^{4/3}n)$, w.p.~$1-O(n^{-c+1})$.

We start the analysis of the path buying phase by bounding the size of $\calC$.
A node $v\in V$ is added to $\calC$ w.p.\ $\frac{c}{n^{1/3}\log^{1/3}n}$,
so $\E[\size{\calC}] = \frac{cn^{2/3}}{\log^{1/3} n}$.
A Chernoff bound implies that
\[\Pr\left[\size{\calC} > \frac{4cn^{2/3}}{\log^{1/3} n}\right]
    \leq \exp\left(- \frac{cn^{2/3}}{\log^{1/3} n}\right) =o(n^{-c}).\]
Similarly, for each value of $k$,
we have $\E[\size{S_k}] = \frac{8c^2 n^{2/3}\log^{2/3} n}{k}$, and
\[\Pr\left[\size{S_k} > \frac{32c^2 n^{2/3}\log^{2/3} n}{k}\right]
\leq \exp\left(-\frac{8c^2 n^{2/3}\log^{2/3} n}{k}\right)
=O(n^{-c}),\]
where the last equality follows since $k\leq \frac{n^{2/3}}{\log^{1/3}n}$.
A union bound implies that
$\size{\calC} = O\left(\frac{n^{2/3}}{\log^{1/3} n}\right)$
and $\size{S_k} = O\left(\frac{n^{2/3}\log^{2/3} n}{k}\right)$
for all $k$,
w.p.\ at least~$1-O(n^{-c+1})$.

Finally, for each $k$,
for each $(c_i,c_j)\in \calC\times S_k$
we add at most one path with less than $2k$ missing edges to $H$.
Thus, for each value of $k$ we add less than
$\size{\calC} \cdot \size{S_k} \cdot 2k = O(n^{4/3}\log^{1/3} n)$
edges to $H$, w.p.\ at least~$1-O(n^{-c+1})$.
Summing over all $O(\log n)$ values of $k$,
and adding the number of edges contributed by the clustering phase,
we conclude that $H$ has at most $O(n^{4/3}\log^{4/3} n)$ edges,
w.p.\ at least~$1-O(n^{-c+1})$.
\end{proof}

\begin{lemma}
\label{lemma: 6ap stretch}
The graph $H$ constructed by \sixap{}
satisfies $\delta_H(x,y) \leq \delta_G(x,y) +6$
for each pair $(x,y)\in V\times V$,
with probability at least~$1-O\left(n^{-c+2}\right)$.
\end{lemma}
\begin{figure}
    \begin{center}
    	\hspace{-8mm}
      \includegraphics[scale=0.92]{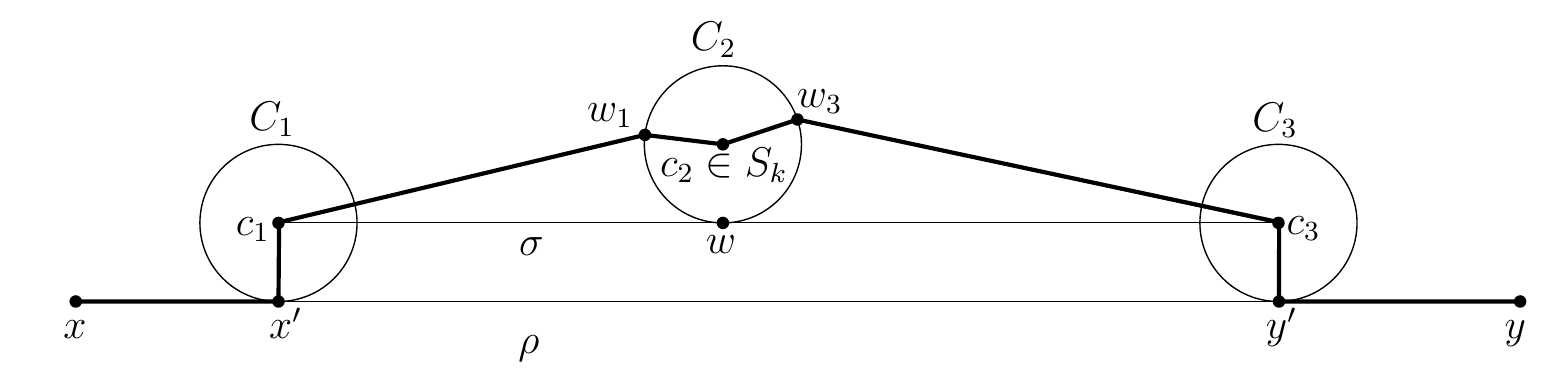}
      \caption{Illustration of the proof of Lemma~\ref{lemma: 6ap stretch}}
      \label{fig: +6 stretch}
    \end{center}
\end{figure}
\begin{proof}
%
Consider a shortest path $\rho$ in $G$ between two nodes $x,y\in V$
(see Figure~\ref{fig: +6 stretch}).
Let $x'$ and $y'$ be the first and last clustered nodes on $\rho$, respectively.
If all nodes of $\rho$ are unclustered,
then $\rho$ is fully contained in $H_0$ and we are done.

Let $c_1$ and $c_3$ be the centers of the clusters
containing $x'$ and $y'$, respectively.
Let $\sigma$ be a shortest path in $G$ between $c_1$ and $c_3$,
and denote by $k'$ the number of edges of $\sigma \setminus H_0$.
Let $k$ be the largest power of $2$ such that $k\leq k'$.

An edge can be in $\sigma \setminus H_0$
only if it connects two clustered nodes.
Hence, $k'$, the number of edges in $\sigma \setminus H_0$,
is smaller than the number of clustered nodes in $\sigma$.
On the other hand,
$\sigma$ cannot contain more than three nodes of the same cluster:
the distance between every two nodes in a cluster is at most two,
so a shortest path cannot traverse more than three nodes of the same cluster.
Thus, the number of clusters intersecting $\sigma$
is at least $k'/3$.
As $k'/3\geq k/3$,
the probability that none of the centers of these clusters
is chosen to $S_k$ is at most
$\left( 1-\frac{8c^2 \log n}{k} \right)^{k/3}
= O\left(n^{-c^2}\right)$.
For each pair of nodes, a cluster center on a shortest path between them is chosen to $S_k$, for the appropriate value of $k$, with similar probability.
A union bound implies that this claim holds for all pairs in $V\times V$
w.p.\ at least~$1-O(n^{-c^2+2})$.

Let $w$ be a node on $\sigma$ in a cluster $C_2$ such that $c_2\in S_k$,
if such a cluster exists.
Denote by $\sigma[c_1,w]$ the sub-path of $\sigma$ from $c_1$ to $w$.
As there are $k' <2k $ edges in $\sigma \setminus H_0$,
there are also less than $2k$ edges in $\sigma[c_1,w] \setminus H_0$.
Thus, in step $3(b)$ of the path-buying phase for $k$,
either the path $\sigma[c_1,w]$ or some other path between $c_1$ and $w$ of length at most
$\delta_G(c_1,w)$ is added to $A$.
In step $3(c)$,
a path from $c_1$ to some node $w_1\in C_2$ is added to $H$,
and this is a shortest path in $A$,
so $\delta_H(c_1, w_1) \leq \delta_G(c_1,w)$.
Similarly, a shortest path from $c_3$ to some $w_3 \in C_2$ is added to $H$,
and $\delta_H(c_3, w_3) \leq \delta_G(c_3,w)$.

The path $\sigma$ is a shortest path from $c_1$ to $c_3$ in $G$,
so $\size{\sigma} \leq \delta_G(x',y') +2$.
As $\delta_G(c_1,w) + \delta_G(c_3,w) = \size{\sigma}$,
we conclude $\delta_H(c_1,w_1) + \delta_H(c_3,w_3) \leq \size{\sigma} \leq \delta_G(x',y') +2$.

Consider the path from $x$ to $y$ in $H$ composed of
the sub-path of $\rho$ from $x$ to $x'$,
the edge $(x',c_1)$,
the path from $c_1$ to $w_1$,
the edges $(w_1,c_2)$ and $(c_2, w_3)$,
the path from $w_3$ to $c_3$,
the edge $(c_3,y')$,
and finally,
the sub-path of $\rho$ from $y'$ to $y$.
This is a path from $x$ to $y$ in $H$, implying
\begin{equation*}
\begin{split}
\delta_H(x,y)
& \leq
\delta_H (x,x') + 1 +\delta_H (c_1,w_1) +2 +
\delta_H (w_3,c_3) +1 +\delta_H (y',y)\\
& \leq
\delta_G (x,x') +4 +\delta_G(x',y') +2 + \delta_G (y',y)
  =
\delta_G(x,y) +6,
\end{split}
\end{equation*}
as desired.
\end{proof}

We now discuss the implementation of \sixap{}
in the \cgst{} model.

\newcommand{\LemmaSixAPComplexity}{
\sixap{} can be implemented in $O\left(\frac{n^{2/3}}{\log^{1/3}n} +D\right)$ rounds
in the \cgst{} model,
with probability at least~$1-o(n^{-c})$.
}

\begin{lemma}
\label{lemma: 6ap complexity}
\LemmaSixAPComplexity
\end{lemma}

\begin{proof}
For the clustering phase,
each node decides locally w.p.\ $\frac{c}{n^{1/3}\log^{1/3}n}$ to become a cluster center,
and notifies its neighbors.
Each node with a neighbor that is a cluster center now joins a cluster
by sending a message to such a neighbor
and adding the appropreate edge to the spanner.
A node with no neighboring cluster centers
notifies all its neighbors and adds all its edges to the spanner.
This is done in a constant number of rounds.

Before the path buying phase,
the nodes construct a single BFS tree,
along which they compute an upper bound $D'$ on $D$,
satisfying $D\leq D'< 2D$,
and count the number of cluster centers, $\size{\calC}$.
The nodes mark the edges of $H_0$ with weight $0$
and the other edges with weight~$1$.
Then, they construct a WBFS tree rooted at each cluster center
by executing Algorithm~\ref{alg: parallel weighted bfs}
for $\size{\calC}+D'$ many rounds.
By the proof of Lemma~\ref{lemma: 6ap size},
we have $\size{\calC}\in O\left(\frac{n^{2/3}}{\log^{1/3}n}\right)$
w.p.\ at least~$1-o(n^{-c})$,
and thus the construction of the WBFS trees takes
$O\left(\frac{n^{2/3}}{\log^{1/3}n} +D\right)$ rounds
with the same probability.

Each node $v$ now knows about a ``good'' path to each cluster center $c_i$, i.e.,
a shortest path from $c_i$ to $v$,
with a minimal number of edges not in $H$ after the clustering phase.
A node $v$ in a cluster $C_j$ notifies its neighbor $c_j$
about all the distances to other cluster centers in $\calC$
and the number of missing edges in each such path.
That is, each $v\in C_j$ sends $\size{\calC}$ messages to $c_j$,
which takes $O\left(\frac{n^{2/3}}{\log^{1/3}n}\right)$ rounds.

Each cluster center $c_j$ decides locally to join each set $S_k$
w.p.\ $\frac{8c^2\log n}{k}$.
For each other center $c_i\in \calC$,
$c_j$ locally constructs the list $A$:
for each $v\in C_j$,
$A$ contains the shortest path from $c_i$ to $v\in C_j$
found by the WBFS algorithm,
and the number of missing edges in it.
Then, $c_j$ chooses from $A$ a path from $c_i$ to some $v\in C_j$
with a minimal number of missing edges,
and if it has at most $2k$ missing edges,
$c_j$  sends a ``buy $c_i$'' message to $v$.

Finally,
all nodes simultaneously execute a ``buy'' phase,
where ``buy $c_i$'' messages are sent up the WBFS tree.
To avoid congestion,
we assume that during the execution of
Algorithm~\ref{alg: parallel weighted bfs},
each node keeps a record of the messages it got in each round
and the WBFS source each message referred to.
Each node $v$ then sends messages in reversed order:
if $v$ has a message ``buy $c_i$'',
and it got a message from $u$ regarding $c_i$ in the $r$-before-last round
of Algorithm~\ref{alg: parallel weighted bfs},
then it sends the message ``buy $c_i$'' to $u$ in round $r$ of the ``buy'' phase.
Then, $u$ adds ``buy $c_i$'' to its list of messages,
and adds the edge $(u,v)$ to the spanner.
This parts takes $O\left(\frac{n^{2/3}}{\log^{1/3}n} +D\right)$~rounds,
just like the execution of Algorithm~\ref{alg: parallel weighted bfs}.
\end{proof}

\section{Discussion and Open Questions}
\label{sec:discussion}
While we present an application of WBFS trees,
our algorithm also solves the weighted $(S,d,k)$-detection problem,
a result that could be of independent interest.

The question of finding the \emph{lightest} paths between all pairs of nodes in a graph, or computing their weights, is a fundamental question in many computational models.
In the \cgst{} model, a randomized algorithm for exactly computing these distances in $\tilde{O}(n)$ rounds was recently presented~\cite{BernsteinN19}.
The exact time complexity of computing these distances \emph{deterministically} in the \cgst{} model is still open~\cite{AgarwalR20}, and we hope our study of lightest shortest paths could facilitate future research on it.

While this paper settles the question of constructing sparse $(+6)$-spanners fast,
the study of spanner construction in distributed environments still lags behind the study of sequential spanner construction algorithms.
In the field of purely additive spanners,
we still do not have fast algorithms,
e.g., for the construction of sparse $(+0)$-pairwise spanners (a.k.a.\ pairwise preservers) and $(+6)$-pairwise spanners.

A more intriguing question
is proving time lower bounds for the construction of spanners in the \cgst{} model:
while $\Omega(D)$ rounds are known to be necessary~\cite{Pettie10},
lower bounds that depend on other parameters of the graph or the spanners
exist only for pairwise spanners~\cite{Censor-HillelKP18}.
Finding a lower bound for the construction of all-pairs spanners in the \cgst{} model, which does not depend on $D$, is still an open question.
Such a lower bound could show that the $n^{3/2}$ term in the time bound of our construction is inevitable,
or motivate the design of faster algorithms for the problem.

\paragraph{Acknowledgements}
We thank Shiri Chechik and Pierre Fraigniaud
for discussions regarding $(+6)$-spanners,
and the reviewers of TCS journal for discussions on consistent shortest paths.
This project has received funding from the European Union's Horizon 2020 Research And Innovation Program under grant agreement no.755839, and from the Israel Science Foundation (grant 1696/14). Ami Paz was supported by the Fondation Sciences Math\'ematiques de Paris (FSMP).

\let\OLDthebibliography\thebibliography
\renewcommand\thebibliography[1]{
	\OLDthebibliography{#1}
	\setlength{\parskip}{0pt}
	\setlength{\itemsep}{0pt plus 0.3ex}
}

\bibliographystyle{plain}
\bibliography{Bibliography}
\end{document}